\documentclass[12pt]{article}
\usepackage{amssymb}
\usepackage{amsmath,array}
\usepackage{amsthm,enumerate}
\usepackage{xcolor}
\usepackage{url,hyperref}

\usepackage{geometry} 
\usepackage{rotating} 

\newtheorem{thm}{Theorem}

\newtheorem{lem}{Lemma}
\newtheorem{cor}{Corollary}

\def\Stab{\textsc{Stabilizer}}

\def\F{\mathbb{F}}

\title{Constructing Permutation Arrays from Groups}

\author{Sergey Bereg\thanks{
Department of Computer Science,
Erik Jonsson School of Engineering and Computer Science,
University of Texas at Dallas.}
\and 
Avi Levy\thanks{
Department of Mathematics,
University of Washington.}
\and I. Hal Sudborough$^*$
}

\date{}

\begin{document}
\maketitle

\begin{abstract}
Let $M(n, d)$ be the maximum size of a permutation array on $n$ symbols with pairwise Hamming distance at least $d$. We use various combinatorial, algebraic, and computational methods to improve lower bounds for $M(n, d)$. 
We compute the Hamming distances of affine semilinear groups and projective semilinear groups, and unions of cosets of $AGL(1,q)$ and $PGL(2,q)$ with Frobenius maps to obtain new,
improved lower bounds for $M(n,d)$.  We give new randomized algorithms. We give better lower bounds for $M(n,d)$ also using new theorems concerning the contraction operation. 
For example, we prove a quadratic lower bound for $M(n,n-2)$ for all $n\equiv 2 \pmod 3$ such that $n+1$ is a prime power.
\end{abstract}

\section{Introduction}

Two permutations $\pi$ and $\sigma$ on $n$ symbols have 
{\em Hamming distance} $hd(\pi,\sigma)=d$ if, for exactly $d$ distinct elements $x$, 
$\pi(x)\neq \sigma(x)$. 
The {\em Hamming distance} of a permutation array $A$, denoted by $hd(A)$, is
the minimum $hd(\pi,\sigma)$ for all permutations $\pi\neq\sigma$ in $A$.
Arrays of permutations on $n$ symbols and with Hamming distance
at least $d$ between any two permutations in the array have been used
for error correcting codes in communication over very noisy power line
channels \cite{h-pc-06,pvyh03}. 
For positive integers $n$ and $d$, with $d\le n$, denote by $M(n,d)$
the maximum size of such an array. 
Constructing maximum size permutation arrays is difficult and, except
for special cases of $n$ and $d$, work has generally been limited to
finding good upper and lower bounds. 
It is known that sharply $k$-transitive groups $G$ of permutations on
$n$ symbols form a maximum size permutation array for pairwise Hamming
distance $d=n-k+1$ \cite{fd77}. 
Except for the exceptional case, i.e. $4$- and $5$-transitive
Mathieu groups \cite{cameron99,conway85,dixon96}, sharply
$k$-transitive permutation groups on $n$ symbols are only known for
$k=2$ and $k=3$ and for $n$ a power of a prime and one more than a
power of a prime, respectively, \cite{pass68}. 
Furthermore, it is known that sharply $k$-transitive groups do not
exist otherwise. 
Combinatorial arguments \cite{ccd04,gao13} are known, which give upper and lower
bounds for $M(n,d)$, for all $n$ and $d$. There are also computational
approaches that have been used to construct good permutation arrays
for small values of $n$. 
Due to the growth rate of $n!$, these computational approaches have
generally been limited to small values of $n$ and the use of
automorphism groups to factor the set of all permutations into a
smaller space and hence extend the range of computational processes
\cite{ccd04,jlos-15,sm12}. 
It is also known that $M(n,n-1)\ge kn$, where $k$ is the number of
mutually orthogonal latin squares (MOLS) of order $n$ \cite{ckl04}, which
means computing large collections of MOLS is related to searching for
large permutation arrays. 
Other techniques that have been used include permutation polynomials
\cite{ccd04}, and special groups, such as the Mathieu groups $M_{22}$, 
$M_{23}$ and $M_{24}$ \cite{cameron99}. 

Let $n$ be a positive integer and $Z_n=\{0,1,2,\dots,n-1\}$. 
Let $S_n$ be the symmetric group on $Z_n$  with composition defined
by $(\pi \sigma)(i) = \sigma(\pi(i))$. 
If $G$ is a group, then 
$G\sigma=\{g\sigma ~|~g\in G\}$ and $\sigma G=\{\sigma g~|~g\in G\}$
are called a {\em right coset} of $G$ and a {\em left coset} of $G$,
respectively, with the {\em representative} $\sigma$ \cite{book}. 

Our results, giving infinitely many improved lower bounds are obtained
by describing collections of cosets of groups. 
We show that groups $A\Gamma L(1,n=p^k)$ and $P\Gamma L(2,n=p^k)$,
where $p$ is prime, obtained from the affine general linear group $AGL(1,n)$
and the projective general group $PGL(2,n)$ by adding $k$ Frobenius
endomorphisms, both have pairwise Hamming distance $n-p^{k^*}$, where
$k^*$ denotes the largest proper\footnote{A factor of $k$ is {\em proper} if it is smaller than $k$.} factor of $k$.   
For example, $A\Gamma L(1,n=2^k)$, when $k$ is prime, consists of
$kn(n-1)$ permutations on $n$ symbols and has pairwise Hamming
distance $2^k-2$, so $M(n,n-2)\ge kn(n-1), P\Gamma L(2,n=2^k)$,
when $k$ is prime, consists of $k(n+1)n(n-1)$ permutations on $n+1$
symbols and has pairwise Hamming distance $2^k-2$, so $M(n+1,n-2) \ge
k(n+1)n(n-1)$. 
We also show that the groups $AGL(1,n)$ and $PGL(2,n)$, where $n=p^k$, together with
the coset defined by the single Frobenius endomorphism $f(x) = x^p$,
has pairwise Hamming distance $n-p$. 

Using a coset technique and computations involving random choices,
we give several improved lower bounds for $M(n,d)$,  For example, we
show that 
$M(13,5) \ge 6,639,048$ and 
$M(14,5) \ge 58,227,624$, 
which improve previous lower bounds.

\section{New Theoretical Bounds}

In this section we use group-theoretic techniques to obtain lower bounds on the sizes of permutation arrays with relatively large Hamming distance. 
First we consider permutation arrays obtained from semilinear groups over finite fields. 
Then we investigate a general procedure of modifying permutation arrays called contraction, and prove sharp Hamming distance bounds which involve the cycle structure of permutations in the permutation array.

Let $\F_n$ be a field of order $n$. Then $n=p^k$, where $p$ is a prime number throughout the remainder of the paper. A polynomial $f(x)\in \F_n[x]$ is called a {\em permutation polynomial} if it is injective. In this case, $f$ permutes the elements of $\F_n$.

We consider the affine general semilinear group $A\Gamma L(1,n)$ and the projective general semilinear group $P\Gamma L(2,n)$. 
These groups are described in the appendix.
It is well known that, for any permutations $\tau,\sigma,\rho\in S_n$,
$hd(\sigma\tau,\sigma\rho)=hd(\tau,\rho)$.

The following lemma shows that computing the Hamming distance of a group $G$ does not require computing the Hamming distance between all the pairs of element in $G$.
It shows that the computation time is $O(|G|)$ rather than $O(|G|^2)$.

\begin{lem}\label{lem:id}
Let $G$ be a group of permutations $G\subseteq S_n$ with $|G|>1$. 
The Hamming distance of $G$, i.e. $hd(G)$, is equal   
$hd(e,G\setminus \{e\})$, where $e$ is the identity permutation.
\end{lem}

\begin{proof}
Pick any two distinct permutations $a,b\in G$.  
Then, $hd(a,b)=hd(e,a^{-1}b)$. Since $a^{-1}b$ is in $G$ and is not the identity, 
the result follows.   
\end{proof}

\begin{lem}\label{lem:cosets}
Let $G$ and $H$ be subgroups of $S_n$, such that $G=\cup_{0\le i\le r} a_iH$, 
for some $r>0$, where $a_0=e$. Then, 
$$hd(G)=\min(\{hd(a_i,H)~|~ 0<i\le r\},hd(e, H\setminus \{e\})).$$
\end{lem} 

\begin{proof} 
Observe that, since $G=\cup_{0\le i\le r} a_iH$, then $G=\cup_{0\le i\le r} H a_i^{-1}$.
Indeed, if $g\in G$, then $g=a_ih$ for some $0\le i\le r$ and $h\in H$.
Then $g^{-1}\in G$ and $g^{-1}=h^{-1}a_i^{-1}\in H a_i^{-1}$. 
So, 
\begin{align*}
hd(G) &= \min(\{hd(e,a_iH)~|~ 0<i\le r\},hd(e, H\setminus \{e\}))\\
&= \min(\{hd(e,H a_i^{-1})~|~ 0<i\le r\},hd(e, H\setminus \{e\})) \\
&= \min(\{hd(a_i,H)~|~ 0<i\le r\},hd(e, H\setminus \{e\})).
\end{align*}
\end{proof}

\begin{thm} \label{hd-agl1}
The Hamming distance of $A\Gamma L(1,n)$ is $n-p^{k^*}$, 
where $n=p^k$ and $k^*$ denotes the largest proper factor of $k$.
\end{thm}

Since $A\Gamma L(1,n)$ contains $kn(n-1)$ elements, we conclude
the following. 

\begin{cor} \label{hd-agl}
Let $n=p^k$ and let $k^*$ be the largest proper factor of $k\ge 1$. 
Then $M(n,n-p^{k^*})\ge kn(n-1)$. In particular, $$M(n,n-2)\ge kn(n-1),\qquad n=2^k,$$ where $k$ is prime. For example, $M(2048,2046)\ge 11\cdot 2048\cdot 2047=46114816$. 
\end{cor}

\begin{cor} \label{hd-agl2}
Let $n=p^k$ for any $k\ge 1$.
Let $s$ be the smallest prime factor of $k$. 
Then $M(n,n-p)\ge s\cdot q(q-1)$. For example,
\begin{align*}
  M(16,14)&\ge 2\cdot 16\cdot 15=480\\
  M(64,62)&\ge 2\cdot 64\cdot 63=8064\\ 
  M(81,78)&\ge 2\cdot 81\cdot 80=12960\\
  M(256,254)&\ge 2\cdot 256\cdot 255=130560\\
  M(512,510)&\ge 3\cdot 512\cdot 511=784896.
\end{align*}
\end{cor} 

\begin{thm}\label{hd-pgl}
The Hamming distance of $P\Gamma L(2,n)$ is $n-p^{k^*}$, where $k^*$ denotes 
the largest proper factor of $k$. 
\end{thm}

Since $P\Gamma L(2,n)$ contains $k(n+1)n(n-1)$ elements, we conclude  the following. 

\begin{cor}
Let $n=p^k$ for any $k\ge 1$.
Let $k^*$ be the largest proper factor of $k$. 
Then $M(n+1,n-p^{k^*})\ge k(n+1)n(n-1)$. In particular when $n=2^p$, we have $M(n+1,n-2)\ge p(n+1)n(n-1)$ and therefore
\begin{align*}
M(9,6)&\ge 3\cdot 9\cdot 8\cdot 7=1512\\
M(17,14)&\ge 2\cdot 17\cdot 16\cdot 15=8160\\
M(28,24)&\ge 3\cdot 28\cdot 27\cdot 26=58968\\
M(33,30)&\ge 5\cdot 33\cdot 32\cdot 31=163680.
\end{align*}
\end{cor}

Using the following lemma, one can compute the Hamming distance simply by counting roots of polynomials.

For any polynomial $f$ with coefficients over a finite field $\F_n$, 
let $r(f)$ denote the number of roots of $f$ in $\F_n$.  
If $f$ and $g$ are polynomials that define permutations on $\F_n$, 
then the Hamming distance between the permutations $f$ and $g$ is equal to 
$n$ minus the number of roots in the polynomial $f-g$, i.e. $hd(f,g) = n - r(f-g)$.  
This is easily seen by observing that $f(x) = g(x)$ is equivalent to $f(x)-g(x) = 0$; 
in other words, $x$ is a root of the polynomial $f-g$.  

\begin{lem}\label{lem:poly}
For any distinct polynomials $f,g\in\F_n[x]$ we have $$hd(f,g) = n-r(f-g).$$
\end{lem} 

The number of roots of polynomials of $AGL(1,n)$ and Frobenius mappings is well known but we give a proof for completeness.

\begin{proof}[Proof of Theorem \ref{hd-agl1}]
Let $H=AGL(1,n)$. 
Since $A\Gamma L(1,n)=\cup_{0\le i<k} x^{p^i}H$ and $hd(H)=n-1$, 
by Lemma \ref{lem:cosets} it suffices to prove $hd(x^{p^i},H) \ge n-p^{k^*}$ for all $1\le i<k$.
By Lemma \ref{lem:poly}, it suffices to prove 
$r(x^{p^i}+ax +b)\le p^{k^*}$ for any $1\le i<k$, and $a\ne 0, b\in \F_n$.
Fix $a\ne 0,b\in\F_n$. 
We show that 
\begin{equation} \label{roots}
r(t_{a,b}(x)) \le r(t_a(x)) \le r(t(x)),
\end{equation}
where $t_{a,b}(x)=x^{p^i}+ax+b, t_a(x)=x^{p^i}+ax$, and $t(x)=x^{p^i}- x$.  

If $t_{a,b}(x)$ has no root in $\F_n$, then clearly $r(t_{a,b}(x)) \le r(t_a(x))$.
Suppose that $t_{a,b}(x)$ has a root, say $y_0$.
For any root $y$ of $t_{a,b}(x)$, we have $t_a(y-y_0) = (y-y_0)^{p^i}+a(y-y_0) 
= y^{p^i} -y_0^{p^i} + ay - ay_0 = t_a(y)-t_a(y_0) = 0$, where the second equation 
follows from the property $(a+b)^p = a^p + b^p$ of Frobenius endomorphisms \cite{cameron99}.  
Thus, $y-y_0$ is a root of $t_a$.  Since the mapping $y\to (y-y_0)$ is an injection, 
the first inequality of (\ref{roots}) follows.

We prove the second inequality of (\ref{roots}), by showing $r(s_a(x))\le r(s(x))$, where $s_a(x)=t_a(x)/x = x^{p^i-1}+a$, 
and $s(x)=t(x)/x = x^{p^i-1}-1$.  
Suppose that $a=0$. Then 0 is the only root of $s_a(x)$.
Since $s(x)$ also has the root 1, the inequality is trivially true.
So, assume that $a\ne 0$. Then 0 is not a root of $s_a(x)$. 
We may also assume that $s_a(x)$ has at least one root; otherwise, the inequality is trivially true.  
Let $z_0$ be a root of $s_a(x)$.  
As $z$ ranges over all roots of $s_a(x)$, map $z$ to $z/z_0$. 
Observe that:
\[ 
s\left(\frac {z}z_0\right)=\left(\frac {z}z_0\right)^{p^i-1}-1=
\frac{z^{p^i-1}}{z_0^{p^i-1}}-1=\frac{-a}{-a}-1=0.
\]
So, $z/z_0$ is a root of $s(x)$. 
Since the map is injective, it follows that $r(t_a(x)) \le r(t(x))$.

Let $S$ be the set of all roots of $t(x) = x^{p^i}- x$.  
Observe that $S$ forms a finite field, since the set of roots are closed under 
the operations of addition, multiplication, and division.  
Thus, $S$ is a subfield of $\F_n$, and hence the cardinality of $S$ divides the cardinality of $\F_n$,
which is $p^k$.  So, $|S| = p^j$, for some $j$, where $j~|~k$.  
Now consider the extension of $t(x)=x^{p^i}- x$ into its splitting field \cite{d-aa-99,ln-ff-97}.  
In this field, the expanded root set forms $\F_{p^i}$.  
So, $S$ is a subfield of $\F_{p^i}$, and so $j~|~i$.  
Thus, $j$ divides both $i$ and $k$, i.e. $j = r(t(x))\le p^{gcd(i,k)}\le p^{k^*}$.  
\end{proof}

\begin{proof}[Proof of Theorem \ref{hd-pgl}]
From the definition of $P\Gamma L(2,n)$ we have 
\[
f(x)= \left\{
\begin{array}{lll}
\frac{ax^{p^i}+b}{cx^{p^i}+d} & \mbox{if } 
x\in\F_n \mbox{ and }cx^{p^i}+d\ne 0,\\
\infty & \mbox{if } 
x\in\F_n, cx^{p^i}+d=0 \mbox{ and }ax^{p^i}+b\ne 0,\\
a/c & \mbox{if } 
x=\infty\mbox{ and }c\ne 0,\\
\infty & \mbox{if } 
x=\infty,c=0 \mbox{ and }a\ne 0.
    \end{array}
  \right.
\]

It follows that $S=\Stab_{P\Gamma L(2,n)}(\infty)$ is isomorphic to $A\Gamma L(1,n)$.\footnote{Recall that,
for a permutation group $G$ on a set $X$, the {\em stabilizer} of an element $x\in X$ is the set of permutations $\{g\in G : g(x)=x\}$.}
So, $hd(S)=hd(A\Gamma L(1,n))=n-p^{k^*}$ by Theorem \ref{hd-agl1}.
Observe that $P\Gamma L(2,n)=\cup_{k=0}^{n-1} \pi_kS$, where $\pi_k$ ($1\le k<n$)
is a permutation in $P\Gamma L(2,n)$ that maps $k$ to $\infty$, and $\pi_0$ is the identity permutation.
(Such permutations are in $P\Gamma L(2,n)$, because it is sharply 3-transitive and hence 2-transitive.)
By Lemma \ref{lem:cosets}, 
\[ hd(P\Gamma L(2,n))=\min(\{hd(\pi_i,S)~|~1\le i<n\},\{hd(e,S-\{e\})). \]
As $hd(\pi_i,S)=hd(S)=hd(A\Gamma L(1,n))\ge n-p^{k^*}$, the result follows. 
\end{proof}

\subsection{Contraction}

In this section we describe a general method of modifying a permutation array called {\em contraction}. 
Contraction allows one to transfer a permutation array from $S_n$ to $S_{n-m}$ without affecting the Hamming distance by too much. 
First, we explain the idea of contraction for $m=1$.

The {\em contraction} of a permutation $\sigma$ on $S_n$, denoted by $\sigma^{CT}$, is the permutation on $S_{n-1}$ defined by the following, where $0\le x\le n-2$:

\begin{equation}
\sigma^{CT}(x) = 
\begin{cases} \sigma(x) & \mbox{if } \sigma(x)\ne n-1, \mbox{ and} \\
\sigma(n-1) & \mbox{if } \sigma(x)=n-1.
\end{cases} 
\end{equation}

That is, the contraction $\sigma^{CT}$ of $\sigma$ is formed by substituting $\sigma(n-1)$  for $n-1$ and deleting the symbol $n-1$ altogether. 
For instance, if $\sigma(n-1)=n-1$, then $\sigma^{CT}$ is formed by simply deleting the symbol $n-1$. 
For a PA $A$ on $S_n$, let $A^{CT}=\{\sigma^{CT}~|~\sigma\in A\}$.  
In general, as the contraction operation leaves most of any permutation’s values untouched, the Hamming distance of $A^{CT}$ (as we shall see) is at least as large as three less than the Hamming distance of $A$ itself (see also \cite{ycy-08}). 
Specifically, a decrease of three in the Hamming distance between two permutations $\rho$ and $\sigma$ occurs if and only if, for some integers $i,j<n-1$, and symbols $r,s<n-1$: (a) $\rho(i)=n-1$ and $\sigma(i)=r$, (b) $\rho(j)=s$ and $\sigma(j)=n-1$, and (c) $\rho(n-1)=r$ and $\sigma(n-1)=s$.  
It follows that there is a decrease of three in the Hamming distance between $\rho$ and $\sigma$ if and only if the permutation $\rho^{-1}\sigma$ contains the cycle $(n-1\ r\ s)$ of length 3. 
If $A$ is a group, since the order of a 3-cycle is three, the order of the group must be divisible by 3 
(by Lagrange's theorem \cite{cameron99}). 
So, if $A$ is a group whose order is not divisible by 3, there can be no 3-cycle and, therefore, the contraction operation decreases the Hamming distance by at most 2.  

\[
\begin{array}{r ccc ccc}
{\rm positions} & \dots & i &\dots & j &\dots & n-1 \\
\rho= & \dots & n-1 & \dots & s & \dots & r \\
\sigma= & \dots & r & \dots & n-1 & \dots & s \end{array} 
\]

Of course the contraction operation can be applied iteratively. 
One can first contract to a permutation on $S_{n-1}$, then $S_{n-2}$, then $S_{n-3}$, $\dots$  
We investigate conditions on the group $A$ that is contracted in order to understand the number of times successive contractions decrease the Hamming distance by 2. 

\begin{lem}\label{lem:2}
Let $\sigma$ and $\tau$ be two permutations in $S_n$. 
If $hd(\sigma^{CT},\tau^{CT})=hd(\sigma,\tau)-2$ then 
(i) $n-1$ is in a cycle $C$ of length at least two in the cycle decomposition of $\sigma^{-1} \tau$, and 
(ii) the cycle decomposition of $(\sigma^{CT})^{-1} \tau^{CT}$ is the same except that the length of $C$ is decreased by two by removing $n-1$ and either the next or the previous element in $C$.
\end{lem} 

\begin{proof}
(i) follows from the fact that, if $\sigma$ and $\tau$ have $n-1$ in the same position, then 
$hd(\sigma^{CT},\tau^{CT})=hd(\sigma,\tau)$.

(ii) If one of the permutations, say $\sigma$, has $n-1$ in position $n-1$, then 
$\sigma(i)=\tau(n-1)$ where $i$ is the position of $n-1$ in $\tau$. Then $C=(\sigma(i),n-1)$ 
and it is deleted in the cycle decomposition of $(\sigma^{CT})^{-1} \tau^{CT}$. 

\[
\begin{array}{r c ccc cc}
{\rm positions} &\dots & i &\dots & n-1 \\
\sigma= & \dots & s & \dots & n-1 \\
\tau= & \dots  & n-1 & \dots & s \end{array} 
\]

Suppose that $\sigma(n-1)\ne n-1\ne \tau(n-1)$ as shown below.
\[
\begin{array}{r ccc ccc}
{\rm positions} & \dots & i &\dots & j &\dots & n-1 \\
\sigma= & \dots & n-1 & \dots & a & \dots & b \\
\tau= & \dots & c & \dots & n-1 & \dots & d \end{array} 
\]
Since $hd(\sigma^{CT},\tau^{CT})=hd(\sigma,\tau)-2$, either $b=c$ or $a=d$. 
Then $C$ is $(a\ n-1\ b\ d\dots)$  and $(b\ a\ n-1\ c\dots)$ in these cases. 
The corresponding cycles in $(\sigma^{CT})^{-1} \tau^{CT}$ are $C-\{n-1,b\}$ and $C-\{a,n-1\}$.
The lemma follows. 
\end{proof}

\begin{thm}\label{contract} 
Suppose a permutation array $P\subset S_n$ has Hamming distance $d$. 
Let $Q\subseteq S_{n-2}$ denote the permutation array obtained from $P$ by applying 
the contraction operation two times.
  \begin{enumerate}[{\rm (a)}]
    \item The Hamming distance of $Q$ is at least $d-6$.
    \item Suppose, for any $\sigma,\tau\in P$, the cycle decomposition of $\sigma^{-1} \tau$ contains no 3-cycle and no 5-cycle. Then the Hamming distance of $Q$ is at least $d-4$.
  \end{enumerate}
\end{thm}

\begin{proof}
The part (a) follows form the fact that the Hamming distance decreases by at most three for each for each contraction operation. 

The part (b) for $k=1$ follows from the above discussion. We prove part (b) for $k=2$. 
Suppose to the contrary that $hd(Q)\le d-5$. 
Let $\sigma$ and $\tau$ be two permutations of $P$ such that 
$hd(\sigma'',\tau'')\le d-5$ where $\sigma''=(\sigma')^{CT},\tau''=(\tau')^{CT}$ and 
$\sigma'=\sigma^{CT},\tau'=\tau^{CT}$.
Since $\sigma^{-1} \tau$ contains no 3-cycle, $hd(\sigma',\tau')\ge d-2$.  
Then the contraction operation on $\sigma'$ and $\tau'$ decreases the Hamming distance by exactly three. Therefore  $(\sigma')^{-1} \tau'$ contains a 3-cycle, say $(a,b,n-1)$.
\[
\begin{array}{r ccc ccc}
{\rm positions} & \dots & i &\dots & j &\dots & n-1 \\
\sigma'= & \dots & b & \dots & n-1 & \dots & a \\
\tau'= & \dots & n-1 & \dots & a & \dots & b \end{array} 
\]

Since  $\sigma^{-1} \tau$ contains no $3$-cycle and $hd(\sigma',\tau')=d-2$,  by Lemma \ref{lem:2}, $\sigma^{-1} \tau$ contains a 5-cycle. Contradiction.
\end{proof}

\begin{cor}  \label{cor4}
(i) For each prime power $n$ such that $n=2\pmod 3$, $$M(n-1,n-3)\ge n(n-1).$$
(ii) For each prime power $n$ such that $n=2\pmod 3$ and $n\neq 0,1\pmod 5$, 
$$M(n-2,n-5)\ge n(n-1).$$
\end{cor}

\begin{proof}
The bound can be obtained by applying the contraction operation on 
the affine general linear group $AGL(1,n)$. Let $\sigma$ and $\tau$ be any 
two distinct permutations from $AGL(1,n)$. Since $\sigma^{-1}\tau\in AGL(1,n)$ and
the order of $AGL(1,n)$ is not multiple of 3, the bound of (i) follows. 
The bound of (ii) follows from Theorem \ref{contract} (b) and the fact that 
the order of $AGL(1,n)$ is not multiple of 3 and 5.
\end{proof}

Infinitely many new bounds can be obtained from Corollary \ref{cor4} (c).  
We show some examples. 
\begin{align*}
M(31,29)&\ge 32\cdot 31=992 \text{ \quad for $n=32$},\\
M(40,38)&\ge 41\cdot 40=1640 \text{ \quad for $n=41$},\\
M(46,44)&\ge 47\cdot 46=2162 \text{ \quad for $n=47$}.
\end{align*}

\begin{thm}  \label{cor:mod3b}
For each prime power $q$ such that $q=2\pmod 3$, $$M(q,q-3)\ge (q+1)q(q-1).$$
\end{thm}

\begin{proof}
Let $A_q=(PGL(2,q+1))^{CT}$. It suffices to show that $hd(A_q)\ge q-3$.
Consider two permutations $\sigma$ and $\tau$ from $PGL(2,q+1)$.
If $hd(\sigma,\tau)=q+1$ then $hd(\sigma^{CT},\tau^{CT})\ge q-2$.
Suppose that $hd(\sigma,\tau)\le q$. Then $\sigma^{-1}\tau$ has a fixed point $k$.
A point stabilizer of $PGL(2,n)$ is isomorphic to the affine general linear group $AGL(1,q)$.
Since its order is not multiple of 3, $\sigma^{-1}\tau$ has no 3-cycle. Therefore
$hd(\sigma^{CT},\tau^{CT})\ge q-3$ and the claim follows.
\end{proof}

An example: $M(32,29)\ge 33\cdot 32\cdot 31=32736$.

\section{Experimental Results}

Computer searches for arrays of permutations on $Z_n$ with pairwise Hamming
distance d are difficult. The number of all permutations, namely $n!$, is large, even for
small values of $n$. Many previous search methods have used automorphism groups
to factor the space of all permutations into collections of sets of permutations with
considerably smaller cardinality \cite{ccd04,jlos-15,sm12}. 
With this smaller cardinality, one constructs a graph $G(n,d)$, whose nodes 
correspond to sets of permutations, with an edge between two sets $S_1$ and $S_2$ 
if the Hamming distance between permutations in the 
sets is at least $d$ \cite{ccd04,jlos-15,sm12}. 
One then uses a program to find a large clique in $G(n,d)$.

We give a different type of algorithm. We call it the {\em coset method}. 
We start with a group $G$ that forms a good PA for $M(n,d')$, where $d'>d$. 
As exact values for $M(q,q-1)$ arise from sharply 2-transitive groups $AGL(1,q)$, and exact values of $M(q+1,q-1)$ arise from sharply 3-transitive groups $PGL(2,q)$, where $q$ is a power of a prime, one easily finds a group with which to start. 
Also, the exact value of $M(n,n)$, for all $n$, arises from a cyclic group of order $n$. 
So, in fact, there is always a group with which to start. 
We search for a permutation, say $\pi$, at Hamming distance at least $d$ from $G$.
It follows that the entire coset $\pi G$ is at distance at least $d$ from $G$. 
That is, 
\begin{align*} 
hd(\pi G,G)&=\max\{ hd(\pi g_1,g_2)~|~g_1,g_2\in G\}\\
&= \max\{hd(\pi,g_2g_1^{-1}~|~g_1,g_2\in G\}\\
&=\max\{hd(\pi,g)~|~g\in G\},
\end{align*}
since $g_2g_1^{-1}\in G$ by properties of a group. 
So, one has a PA of cardinality $2\cdot |G|$, by finding a single permutation $\pi$. 
Moreover, recording the coset representative $\pi$ is sufficient, as all other 
permutations in the coset are obtained from the group $G$ and $\pi$.

The process can be iterated. Suppose we have found $k$ coset representatives
$\pi_0,\pi_1,\dots,\pi_{k-1}$, where $\pi_0$ is the identity permutation, so $\pi_0 G$ is the group $G$ and the collection of all such cosets is a PA of cardinality $k\cdot |G|$. 
One can continue by finding a permutation $\pi_k$ such that $hd(\pi_k,\pi_iG)\ge d$, for all $0\le i<k$. 
Such a permutation gives us a new coset, namely $\pi_k G$, and a PA of cardinality $(k+1)\cdot |G|$.

One implementation of this method guesses a new permutation $\pi_k$ randomly and
then checks (1) that $hd(\pi_k,\pi_iG)\ge d$, for all $0\le i<k$, and 
(2) $hd(\pi_k,\pi_iG)=d$, for some $i$ ($0\le i<k$).
The second condition may need explanation. 
Recall the combinatorial Gilbert-Varshimov (GV) lower bound \cite{fd77,ms-tecc-77,ph-hct-98} for $M(n,d)$, namely $N_{GV}(n,d)$, 
where $N_{GV}(n,d)$ is given by:
\begin{equation} \label{GV}
N_{GV}(n,d)=\frac{n!}{V(n,d-1)},
\end{equation}
where $V(n,d-1)$, the number of permutations that are at distance less than $d$ from a given
permutation, e.g. the identity, is $V(n,d-1)=\sum_{k=0}^{d-1}{n\choose k} D_k$, where $D_k$ is the number of {\em derangements} on $k$ symbols.
Note that the ratio given in (\ref{GV}) is a lower bound for the number of
times one can select another permutation without choosing two with Hamming
distance less than $d$. 
It is calculated with the assumption that all sets of 
permutations too close to different permutations are disjoint. 
(It should be noted that one can often get far more than $N_{GV}(n,d)$ permutations in a PA, 
because the assumption that the spheres are disjoint is not required. 
In fact, this combinatorial lower bound has recently been improved \cite{gao13}.)
Of course, such sets need not be disjoint and one can choose a larger set of 
permutations by eliminating this condition. 
We do this by requiring each new permutation selected is at distance {\em exactly} $d$
to one chosen before.

This simple technique makes it feasible to compute and verify large PA's. 
Previously, the implicit algorithm justifying the GV lower bound was not 
considered practical, due to large space requirements for keeping track 
of permutations already chosen and the large time needed for computing 
Hamming distances. 
As we have seen in Lemma \ref{lem:id}, when the PA is a group and some 
of its cosets, one need not check the distance between every pair of 
permutations and one can store the set of coset representatives instead of 
the set of all permutations. 
It is worth noting that in \cite{ccd04}, the authors computed a PA of 
size 58,322 for $M(16,9)$ and stated that this lower bound is not as good 
as what is given by the GV lower bound\footnote{Note that $N_{GV}(16,9)=97,579$.}. 
They stated, "... However, the GV lower bound is not constructive." 
On the other hand, the coset method starting with the group $G=AGL(1,16)$ 
found $5,739$ cosets of $G$ for Hamming distance 9 and, hence, obtained 
a lower bound of $1,377,360$ for $M(16,9)$.

Also, in \cite{sm12} the authors computed a PA of size $20,908,800$ for $M(12,4)$ and then
stated "... it is too large to check fully, but has been extensively checked." 
In contrast, the coset method starting with the Mathieu group $G$ of order 12 
\cite{cameron99,conway85,dixon96} found 638 cosets of $G$ for Hamming distance 4 and, hence, obtained a lower bound of $60,635,520$ for $M(12,4)$. 
Furthermore, since testing of correctness of $G$ and its cosets takes far less
time using group-theoretic properties, verification was done in a few minutes using
a computer.

Verifying the Hamming distance of a PA $A$ on $Z_n$ of size $N$ generally involves
computing ${N\choose 2}$ pairs of permutations of $n$ symbols, which is $O(N^2n)$ time. When $A$ is a group, using Lemma \ref{lem:id}, one only need compute 
the distance between the identity and the other $N-1$ elements, so $O(Nn)$ time. 
In fact, if $A$ is a group consisting of the identity permutation and its cyclic shifts or $A$ consists of this cyclic group and, say, $k$ of its cosets, the computation time is reduced to $O(N)$ and $O(kN)$, respectively. We discuss testing algorithms in Section \ref{testing}.

The coset method has been used to obtain several new lower bounds. 
Many of the new lower bounds are given in Tables 1 and 2. 
The PA's that justify the lower bounds are available at our web site \cite{ourtable}. 
Sometimes finding a suitable coset representative takes considerable computation time and, hence, when found it should be recorded. 
For example, when started with the group $G=PGL(2,19)$ the 
coset method with difficulty found one coset of $G$ for Hamming distance 16. 
It can be described by one of its representatives, for example:
$$1,2,3,4,6,8,15,5,19,18,10,7,17,16,12,20,9,11,13,14.$$
Thus, $M(20,16) \ge 13,680$. Some other lower bounds obtained are:

$M(13,5) \ge 10,454,400$

$M(13,6) \ge 1,805,760$

$M(13,7) \ge 380,160$

$M(14,5) \ge 60,445,440$

$M(14,6) \ge 10,834,560$

$M(14,7) \ge 1,900,800$

$M(14,8) \ge 380,160$

$M(14,9) \ge 21,840$

$M(15,6) \ge 58,734,720$

$M(15,7) \ge 15,491,520$

$M(15,8) \ge 1,900,800$

$M(15,9) \ge 181,272$

$M(15,10) \ge 32,760$

$M(16,7) \ge 70,709,760$

$M(16,8) \ge 16,061,760$

$M(16,9) \ge 1,377,360$

$M(16,10) \ge 164,880$.

Others can be found in Tables 1 and 2 and at the website.

\subsection{Efficient Testing of a New Permutation} \label{testing}

Suppose that we have a permutation array $P$ consisting of $k$ left cosets of a group $G$, i.e. 
$P=\cup_{i=1}^k P_i$.
The critical step of our randomized construction is the computation of distance $hd(\pi,P)$ where $\pi$ is a random permutation. 
The definition of Hamming distance suggests the computation  
$hd(\pi,P)=\min_{\sigma\in P} hd(\pi,\sigma)$. 
Then the running time is $O(n|P|)=O(nk|G|)$. 
We show that, if $G$ is cyclic or has the cyclic subgroup, the
algorithm for testing $\pi$ can be improved. Let $C_n$ denote the
cyclic group, i.e. 
$$C_n=\{g_j~|~g_j=(j,j+1,\dots,n-1,0,1,\dots,j-1), j\in Z_n\}.$$ 

\begin{lem} \label{testC1}
The Hamming distance from a permutation $\pi\in S_n$ to $C_n$ can be
computed in $O(n)$ time. 
\end{lem}

\begin{proof}
The Hamming distance $hd(\pi,C_n)$ can be computed as 
\[ hd(\pi,C_n)=\min_{0\le j\le n-1}\{hd(\pi,g_j)\}. \]
The straightforward computation of $hd(\pi,g_j)$ takes $O(n)$ time.
We show that it can be computed in $O(1)$ amortized time.

Let $D[0..n-1]$ be an array and we store a tentative distance of
$hd(\pi,g_j)$ in $D[j]$. 
The algorithm has 2 steps.
\begin{enumerate}
\item 
We initialize $D[j]=n$ for all $j\in Z_n$.
\item 
For each $m\in Z_n$, subtract one from $D[j]$ where $j\equiv (\pi(m)-m)\bmod{n}$.
\end{enumerate}
Clearly, the running time is $O(n)$.
We show that the algorithm is correct.
If $D[j]$ decreases for some $m$, then the $m$-th element of $\pi$ is $j+m(\bmod\ n)$.
Thus, $\pi$ and $g_j$ have matching elements in $m$-th position. 
In the end, $D[j]$ is equal to $n-n'$ where $n'$ is the number of matching elements of 
$\pi$ and $g_j$. The claim follows since $hd(\pi,g_j)=n-n'$.
\end{proof}

\begin{lem} \label{testC}
Suppose that $G=C_n$ and let $\pi_i, (i=1,2,\dots,k)$ be a representative of the coset $P_i$, i.e. $P_i=\pi_iC_n$.
Then, for any permutation $\pi\in S_n$, the Hamming distance $hd(\pi,P)$ can be computed in $O(kn)$ time.
\end{lem}

\begin{proof}
The Hamming distance $hd(\pi,P)$ can be computed as 
\[ hd(\pi,P)=hd(\pi,\mathop\cup_{1\le i\le k}P_i)=
\min_{1\le i\le k}
\{hd(\pi,P_i\}. \]

It suffices to show that $d_i=hd(\pi,P_i)$ can be computed in $O(n)$ time, for any $i$.
Since $P_i=\pi_iC_n$ and $d_i=hd(\pi_i^{-1}\pi, C_n)$, we first compute $\sigma=\pi_i^{-1}\pi$ and then 
$hd(\sigma,C_n)$ using the algorithm from Lemma \ref{testC1}.
\end{proof}

We generalize Lemma \ref{testC} as follows.

\begin{thm}
If $C_n$ is a subgroup of $G$ then the Hamming distance $hd(\pi,P)$
can be computed in $O(k|G|)$ time for any permutation $\pi\in S_n$. 
\end{thm}

\begin{proof}
Let $G/C_n=\{\sigma C_n~|~\sigma\in G\}$ be the set of left cosets of $C_n$ in $G$. 
Let $\sigma_1,\sigma_2,\dots,\sigma_m$ be representatives of these cosets.
The Hamming distance $hd(\pi,P)$ can be computed as 
$$hd(\pi,P)=hd(\pi,\mathop\bigcup_{1\le i\le k}\pi_iG)=
hd(\pi,\mathop\bigcup_{\substack{1\le i\le k\\ 1\le t\le m}}\pi_i\sigma_tC_n)=
\min_{\substack{1\le i\le k\\ 1\le t\le m}}
\{hd(\pi,\pi_i\sigma_tC_n)\}.$$
Since $hd(\pi,\pi_i\sigma_tC_n)=hd(\alpha,C_n)$ where $\alpha=(\pi_i\sigma_t)^{-1}\pi$. 
The Hamming distance $hd(\alpha,C_n)$ can be computed in $O(n)$ time using the algorithm from Lemma \ref{testC1}. 
The total time for computing $hd(\pi,P)$ is $O(kmn)$. It can be written as $O(k|G|)$ since $|G|=mn$.
\end{proof}

\section{New tables and conclusions}

We give in Tables 1 and 2 an updated partial list of lower bounds for $M(n,d)$, 
with $n\ge 9$ and $d\ge 4$.
We also created a webpage \cite{ourtable} that allows one to obtain the PA's, or 
coded versions of them, for verification.
Not all of our new results appear in the table, many are described by theorems in Section 2.

We use the following notation in the table to describe how the results are obtained:

\vspace{-12pt}
\begin{description} \setlength{\itemsep}{-5pt}
\item[{\bf a}] - a bound derived from $M(n,d-1)\ge M(n,d)$.
\item[{\bf b}] - a bound derived from $M(n+1,d)\ge M(n,d)$. 
\item[\bf d] - a bound derived from $M(n-1,d)\ge M(n,d)/n$. 
For example, we have $M(15,12)\ge 2520$, because $M(16,12)\ge 40320$.
\item[{\bf t}] - the Gilbert-Varshamov lower bound \cite{fd77,ms-tecc-77,ph-hct-98}. 
\item[{\bf g}] - a lower bound based on known permutation groups. 
This includes the Mathieu groups $M_{11}, M_{12},M_{22},M_{23},M_{24}$
and groups $AGL(1,n)$ or $AGL(2,n)$, which are sharply 
2-transitive and sharply 3-transitive, respectively. 
\item[{\bf m}] - a bound derived from mutually orthogonal Latin squares \cite{ckl04}. 
\item[{\bf u}] - a result obtained by {\em partitioning and extending}, which is 
contained in \cite{bms-epa-16}. 
\item[{\bf r}] - a result obtained by the coset method using random search. 
Our improved lower bounds are in bold. 
The previous lower bounds are given at the bottom of the cell.
When a lower bound is obtained by the coset technique, the number of cosets is given at the cell.
\item[{\bf c}] - a result obtained by {\em contraction}, which is an 
improvement of the $M(n-1,d-3)\ge M(n,d)$ result, described in \cite{ycy-08}.
For example, $M(23,20)\ge 12144$ comes from $M(24,22)\ge 12144$.
\end{description}
\vspace{-4mm}
We use capital letters for the following references:
A - \cite{ccd04}, B - \cite{gao13}, C - \cite{jlos-15}, D - \cite{js08}, 
E - \cite{sm12}, F - \cite{ycy-08}, G - \cite{bms-psp-16}.

We note that $M(9,6)\ge 1512$ and $M(17,14)\ge 8160$
follow from Corollaries 1 and 2. 

\newpage
\newgeometry{left=3cm,top=0.1cm,bottom=0.1cm}
\begin{table} 
\caption{Bounds for $M(n,d)$, where the rows are labeled with $n$ referring to 
the set $S_n$, which the permutations are defined on, and the columns $d$ are the Hamming distance.}
{\small\sffamily
\rotatebox{90}{
\begin{tabular}{|r|*{13}{@{\hspace{1pt}}r@{\hspace{1pt}}|}}
\hline
 & d=4 & d=5 & d=6 & d=7 & d=8 & d=9 & d=10 & d=11 & d=12 & d=13 & d=14 & d=15 & d=16 \\ \hline 
9 & -& -& -& -& -& -& -& -& -& -& -& -& -\\ 
& -& -& -& -& -& -& -& -& -& -& -& -& -\\ 
& 18576$_{ C }$& 3024$_{ E }$& 1512$_{ E }$& 504$_{ g }$& 72$_{ g }$& 9$_{}$& -& -& -& -& -& -& -\\ \hline 
10 & -& -& -& -& -& -& -& -& -& -& -& -& -\\ 
& -& -& -& {\bf 1504}$_u$& -& -& -& -& -& -& -& -& -\\ 
& 150480$_{ E }$& 19440$_{ C }$& 8640$_{ E }$& 1484$_{ C }$& 720$_{ g }$& 49$_{ D }$& 10$_{}$& -& -& -& -& -& -\\ \hline 
11 & -& -& -& -& -& -& -& -& -& -& -& -& -\\ 
& -& -& -& -& -& -& -& -& -& -& -& -& -\\ 
& 1742400$_{ E }$& 205920$_{ E }$& 95040$_{ E }$& 7920$_{ a }$& 7920$_{ g }$& 297$_{ C }$& 110$_{ g }$& 11$_{}$& -& -& -& -& -\\ \hline 
12 & -& -& -& -& -& -& -& -& -& -& -& -& -\\ 
& -& -& -& -& -& -& -& -& -& -& -& -& -\\ 
& 20908800$_{ E }$& 2376000$_{ E }$& 190080$_{ E }$& 95040$_{ a }$& 95040$_{ g }$& 1320$_{ a }$& 1320$_{ g }$& 112$_{ C }$& 12$_{}$& -& -& -& -\\ \hline 
13 & 638 & 110 & 20 & 4 & -& -& -& -& -& -& -& -& -\\ 
& {\bf 60635520}$_r$& {\bf 10454400}$_r$& {\bf 1900800}$_r$& {\bf 380160}$_r$& -& -& -& -& -& -& -& -& -\\ 
& 41712480$_{ A }$& 2376000$_{ b }$& 271908$_{ t }$& 95040$_{ a }$& 95040$_{ b }$& 6474$_{ C }$& 1320$_{ b }$& 276$_{ C }$& 156$_{ g }$& 13$_{}$& -& -& -\\ \hline 
14 & -& 636 & 115 & 21 & 4 & 10 & -& -& -& -& -& -& -\\ 
& -& {\bf 60445440}$_r$& {\bf 10929600}$_r$& {\bf 1995840}$_r$& {\bf 380160}$_r$& {\bf 21840}$_r$& -& -& -& -& -& -& -\\ 
& 550368000$_{ A }$& 22767826$_{ B }$& 890338$_{ t }$& 97547$_{ t }$& 95040$_{ b }$& 6552$_{ a }$& 8736$_{ C }$& 2184$_{ a }$& 2184$_{ g }$& 59$_{ C }$& 14$_{}$& -& -\\ \hline 
15 & -& -& 618 & 163 & 20 & 83 & 15 & -& -& -& -& -& -\\ 
& -& -& {\bf 58734720}$_r$& {\bf 15491520}$_r$& {\bf 1900800}$_r$& {\bf 181272}$_r$& {\bf 32760}$_r$& -& -& -& -& -& -\\ 
& 4.01E9$_B$& 263832788$_{ B }$& 8991655$_{ t }$& 888533$_{ t }$& 97572$_{ t }$& 12014$_{ t }$& 6076$_{ a }$& 7540$_{ C }$& 2520$_{ d }$& 315$_{ C }$& 90$_{ C }$& 15$_{}$& -\\ \hline 
16 & -& -& -& 744 & 169 & 5739 & 687 & -& -& -& -& -& -\\ 
& -& -& -& {\bf 70709760}$_r$& {\bf 16061760}$_r$& {\bf 1377360}$_r$& {\bf 164880}$_r$& -& -& -& -& -& -\\ 
& 1.268E11$_B$& 3.317E9$_B$& -& 8972298$_{ t }$& 888755$_{ t }$& 97579$_{ t }$& 40320$_{ a }$& 40320$_{ a }$& 40320$_{ E }$& 1376$_{ C }$& 1376$_{ C }$& 240$_{ g }$& 16$_{}$\\ \hline 
17 & -& -& -& -& 791 & 2298 & 303 & 46 & -& 3 & -& -& -\\ 
& -& -& -& -& {\bf 75176640}$_r$& {\bf 9375840}$_r$& {\bf 1236240}$_r$& {\bf 187680}$_r$& -& {\bf 12240}$_r$& {\bf 8160}& -& -\\ 
& 7.93E11$_B$& -& -& -& 8974885$_{ t }$& 888727$_{ t }$& 97569$_{ t }$& 12014$_{ t }$& 83504$_{ A }$& 4080$_{ a }$& 4080$_{ a }$& 4080$_{ g }$& 272$_{ g }$\\ \hline 
\end{tabular}}}
\label{table1} 
\end{table} 

\newpage
\newgeometry{left=3cm,top=0.1cm,bottom=0.1cm}
\begin{table}
\caption{Bounds for $M(n,d)$, where the rows are labeled with $n$ referring to 
the set $S_n$, which the permutations are defined on, and the columns $d$ are the Hamming distance.}
{\small\sffamily
\rotatebox{90}{
\begin{tabular}{|r|*{14}{@{\hspace{1pt}}r@{\hspace{1pt}}|}}
\hline
 & d=11 & d=12 & d=13 & d=14 & d=15 & d=16 & d=17 & d=18 & d=19 & d=20 & d=21 & d=22 & d=23 & d=24 \\ \hline 
18 & 251 & 40 & 5 & -& -& -& -& -& -& -& -& -& -& -\\ 
& {\bf 1228896}$_r$& {\bf 195840}$_r$& {\bf 24480}$_r$& {\bf 12240}$_u$& {\bf 8160}$_u$& -& -& -& -& -& -& -& -& -\\ 
& 97569$_{ t }$& 83504$_{ t }$& 4896$_{ a }$& 4896$_{ a }$& 4896$_{ a }$& 4896$_{ g }$& 90$_{ C }$& 18$_{}$& -& -& -& -& -& -\\ \hline 
19 & -& 3572 & 486 & -& -& -& -& -& -& -& -& -& -& -\\ 
& -& {\bf 1221624}$_r$& {\bf 166212}$_r$& -& {\bf 12240}$_G$& -& -& -& -& -& -& -& -& -\\ 
& 888729$_{ t }$& 97569$_{ t }$& 65322$_{ a }$& 65322$_{ F }$& -& 4896$_{ b }$& 342$_{ a }$& 342$_{ g }$& 19$_{}$& -& -& -& -& -\\ \hline 
20 & -& 1299 & 215 & 28 & 3 & 2 & -& -& -& -& -& -& -& -\\ 
& -& {\bf 8885160}$_r$& {\bf 1470600}$_r$& {\bf 191520}$_r$& {\bf 20520}$_r$& {\bf 13680}$_r$& -& -& -& -& -& -& -& -\\ 
& 8974608$_{ t }$& 888729$_{ t }$& 97569$_{ t }$& 12014& 65322$_{}$& 6840$_{ a }$& 6840$_{ a }$& 6840$_{ g }$& 120$_{ C }$& 20$_{}$& -& -& -& -\\ \hline 
21 & -& -& -& 174 & 23 & 5 & 2 & -& -& -& -& -& -& -\\ 
& -& -& -& {\bf 1190160}$_r$& {\bf 157320}$_r$& {\bf 34200}$_r$& {\bf 13680}$_u$& -& {\bf 333}$_G$& -& -& -& -& -\\ 
& -& -& 888729$_{ t }$& 97569$_{ t }$& 65322$_{ b }$& 6840$_{ a }$& 6840$_{ a }$& 6840$_{ b }$& 147$_{ a }$& 147$_{ C }$& 21$_{}$& -& -& -\\ \hline 
22 & -& -& -& -& 47233 & -& 1250 & -& -& -& -& -& -& -\\ 
& -& -& -& -& {\bf 1039126}$_r$& -& {\bf 27500}$_r$& {\bf 13680}$_u$& {\bf 1100}$_G$& {\bf 528}$_d$& {\bf 104}$_u$& -& -& -\\ 
& -& -& -& -& 443520$_{ a }$& 443520$_{ g }$& 6840$_{ a }$& 6840$_{ b }$& -& 220$_{ A }$& 121$_{ C }$& 22$_{}$& -& -\\ \hline 
23 & -& -& -& -& -& -& -& -& -& -& -& -& -& -\\ 
& -& -& -& -& -& -& -& -& {\bf 12144}$_a$& {\bf 12144}$_c$& -& -& -& -\\ 
& -& -& -& -& 10200960$_{ a }$& 10200960$_{ g }$& 291456$_{ a }$& 291456$_{ A }$& -& -& 506$_{ a }$& 506$_{ g }$& 23$_{}$& -\\ \hline 
24 & -& -& -& -& -& -& 89 & -& 2 & -& -& -& -& -\\ 
& -& -& -& -& -& -& {\bf 1080816}$_r$& -& {\bf 24288}$_r$& -& -& -& -& -\\ 
& -& -& -& -& -& 2.44E8$_g$& 97569$_{ t }$& 291456$_{ b }$& 23782$_{ a }$& 23782$_{ F }$& 12144$_{ a }$& 12144$_{ g }$& 168$_{ m }$& 24$_{}$\\ \hline 
25 & -& -& -& -& -& -& -& -& -& -& -& -& -& -\\ 
& -& -& -& -& -& -& -& -& {\bf 61200}$_d$& {\bf 15600}$_a$& {\bf 15600}$_c$& -& -& -\\ 
& -& -& -& -& -& -& -& -& -& -& -& 12144$_{ b }$& 600$_{ a }$& 600$_{ g }$\\ \hline 
26 & -& -& -& -& -& -& -& -& 102 & 10 & 2 & -& -& -\\ 
& -& -& -& -& -& -& -& -& {\bf 1591200}$_r$& {\bf 156000}$_r$& {\bf 31200}$_r$& -& -& -\\ 
& -& -& -& -& -& -& -& -& -& -& 12144$_{ t }$& -& -& 15600$_{ g }$\\ \hline 
\end{tabular}}}
\label{table2} 
\end{table} 

\restoregeometry

\newpage

{\bf \large Appendix}

{\em Finite Fields}. 
Let $n=p^k$ be a prime power. There is a field $\F_n$ with $n$ elements, unique up to isomorphism. We consider groups over the field $\F_n$.

{\em Groups}.
The group $AGL(1,n)$ consists of the affine linear transformations
\[AGL(1,n) = \{ax + b | a,b\in\F_n, a\not=0\},\]
where the group operation is function composition. 
This group is sharply $2-$transitive and has $n(n-1)$ elements.

Denote the symbols of $Z_{n+1}$ by $0,1,2,\dots,n-1,\infty$. 
The permutations of $PGL(2,n)$ are $g:x\to \frac{ax+b}{cx+d}$ on $Z_{n+1}$
such that $a,b,c,d\in GF(n), ad\ne bc, g(\infty)=a/c,g(-d/c)=\infty$ 
if $c\ne 0$ and $g(\infty)=\infty$ if $c=0$. 
Then $|PGL(2,n)|=(n+1)n(n-1)$. 

Recall that $n=p^k$. The group of affine semilinear polynomials $A\Gamma L(1,n)$ arises as a semidirect product of $AGL(1,n)$ with a cyclic group of order $k$. It is generated by iteratively composing the Frobenius automorphism $x^p$ with the elements of $AGL(1,n)$. Equivalently,
\[A\Gamma L(1,n) = \{ax^{p^i}+b \mid a,b\in\F_n, a\not=0, 0\leq i < k\}\]
This group has $kn(n-1)$ elements.

The group of projective semilinear polynomials $P\Gamma L(2,n)$ arises as a semidirect product of $PGL(2,n)$ with a cyclic group of order $k$ generated by the Frobenius automorphism.
This group has $k(n+1)n(n-1)$ elements.

\end{document}